\documentclass{article}
\usepackage{amsmath,amsfonts,amsthm,setspace}
\usepackage[a4paper,left=3cm, right=3cm, top=3cm,bottom=3cm,nomarginpar, includeheadfoot,head=12pt,headsep=18pt]{geometry}
\usepackage[latin1]{inputenc}
\usepackage{graphicx}

\usepackage{color}

\newtheorem{theorem}{Theorem}
\newtheorem{lemma}{Lemma}
\newtheorem{proposition}{Proposition}
\newtheorem{corollary}{Corollary}
\theoremstyle{definition}

\newtheorem{remark}{Remark}

\newcommand{\ee}{E}

\newcommand{\rr}{\mathbb R}

\newcommand{\dd}{\text d}
\newcommand{\ii}{\text i}

\newcounter{comcount}
\newcommand{\comment}[1]{}
\newcommand{\commenttwo}[1]{}

\title{Tracking errors from discrete hedging in exponential Lévy models.}

\author{Mats Brod\'en \\ Centre for Mathematical Sciences\\ Lund University, 22100 Lund Sweden \\E-mail: \texttt{matsb@maths.lth.se} \and 
Peter Tankov\footnote{Corresponding author} \\ Centre de Math\'ematiques Appliqu\'ees,\\ Ecole Polytechnique, 91128 Palaiseau Cedex France \\E-mail: \texttt{peter.tankov@polytechnique.org}}
\date{}

\begin{document}

\maketitle

\begin{abstract}
We analyze the errors arising from discrete readjustment of the hedging portfolio when hedging options in exponential Lévy models, and establish the rate at which the expected squared error goes to zero when the readjustment frequency increases. We compare the quadratic hedging strategy with the common market practice of delta hedging, and show that for discontinuous option pay-offs the latter strategy may suffer from very large discretization errors.  For options with discontinuous pay-offs, the convergence rate depends on the underlying Lévy process, and we give an explicit relation between the rate and the Blumenthal-Getoor index of the process.  
\end{abstract}




\section{Introduction}
We study the problem of hedging an option with a discretely rebalanced portfolio in an exponential Lévy model. This setting corresponds to an incomplete market and therefore gives rise to two kinds of hedging errors. The market incompleteness error is the difference between the option's pay-off and the theoretical hedging portfolio which assumes continuous rebalancing. This error and its minimization has been analyzed in several papers in the context of exponential Lévy models \cite{kallsen.hubalek.al.06,cont.al.05}. In this study we therefore focus on the discretization error, denoted by $\varepsilon_T$, and defined as the difference between the theoretical continuously rebalanced portfolio and the discretely rebalanced one. 


The error from discrete-time hedging and the related problem of approximating a stochastic integral with a Riemann sum has been analyzed by several authors in the context of diffusion models or continuous Itô processes. Bertsimas, Kogan and Lo \cite{bertsimas.kogan.lo.00} and later Hayashi and Mykland \cite{hayashi.mykland.05} gave the conditions under which the renormalized hedging error $\sqrt{n}\varepsilon_T$ converges weakly to a nondegenerate limiting distribution as the number of discretization dates $n$ goes to infinity. The rate of $L^2$ convergence of the discretization error to zero was analyzed by Zhang \cite{zhang.couverture}, who showed that for European Call and Put options $n E[\varepsilon_T^2]$ converges to a nonzero finite limit as $n\to \infty$ and by Gobet and Temam \cite{gobet.temam.01}, who studied irregular pay-offs and showed in particular that for digital options $\sqrt{n} E[\varepsilon_T^2]$ converges to a nondegenerate limit. Geiss \cite{geiss.02,geiss.geiss.06}, showed that for irregular pay-off functions the convergence rate of $n$ rather than $\sqrt{n}$ may be recovered by taking a non-equidistant (but deterministic) time net, where the rebalancing frequency increases as the option approaches expiry. 

In the context of discontinuous processes, the limiting behavior of the discretization error was studied in \cite{tankov.voltchkova.09} from the point of view of weak convergence, and it was shown in particular that if the underlying process has no diffusion component, $\sqrt{n}\varepsilon_T\to 0$ in probability as $n\to \infty$. However, in financial applications the risk is more commonly measured by an $L^2$ criterion. In this paper we therefore concentrate on the rate of $L^2$ convergence of the discretization error to zero, and we show that for this criterion, the convergence rates are different both from the rates of weak convergence and from the rates found by various authors for continuous processes. 

First, the limit $\lim_{n\to \infty} n E[\varepsilon_T^2]$ is positive in all cases and may be infinite. This means that for pure-jump Lévy processes, the rate of $L^2$ convergence is different from the rate of convergence in probability. This phenomenon is not encountered in diffusion models, and is explained by the fact that the big jumps do not contribute to the rate of convergence in probability, while they do contribute to the rate of $L^2$ convergence. 

Second, the convergence rate of the discretization error may depend on the hedging strategy. In this paper, we analyze and compare two specific hedging strategies: the quadratic hedging, which is the optimal strategy for the $L^2$ criterion, and the delta hedging, which is known to be suboptimal in exponential Lévy models, but is commonly used in practice and has been shown to be relatively close to optimal in terms of hedging error \cite{denkl}. We find that although for options with regular pay-offs the two strategies have similar discretization errors, in the case of irregular pay-offs, the delta hedging strategy, because it involves differentiation of the option price function, suffers from much larger discretization errors than quadratic hedging.

Finally we show that for options with irregular pay-offs, such as digitals, the convergence rate of the discretization error depends on the fine properties of the Lévy measure near zero.  We assume that the small jumps have stable-like behavior with index $\alpha$ (which is the case in many models used in practice) and characterize the convergence rates for the two strategies depending on $\alpha$.

In this paper, we suppose that the rebalancing dates are equidistant. Equidistant rebalancing is common market practice, especially for not-so-liquid underlyings which are not observed continuously. Although the convergence rates for options with irregular pay-offs may be improved by taking non-equidistant dates as in \cite{geiss.02}, in many practical situations such non-equidistant time grids cannot be used (for example, when one needs to hedge a portfolio of options on the same underlying with different expiry dates).  

The rest of the paper is structured as follows.  
After recalling the Fourier transform approach to option pricing in exponential Lévy models in section \ref{pricing.sec}, we establish, in section \ref{general.sec}, a general criterion for the $L^2$ convergence with a given rate of the error from discrete hedging.  In section \ref{specific.sec}, we first study the case of options with regular (Lipschitz) pay-offs and show that in this case for both quadratic hedging and delta hedging maximum convergence rate is attained. Next we turn to options with discontinuous payoffs and compute the convergence rates for both strategies as function of the parameter $\alpha$ characterzing the small jumps. In this case, the convergence rate for delta hedging is found to be strictly lower than for quadratic hedging, which shows that the former strategy may suffer from much larger discretization errors and should therefore be avoided in practice.

\section{Pricing and hedging in exp-Lévy models}
\label{pricing.sec}

\paragraph{Standard notation and basic assumptions} We now introduce the common notation for the rest of the paper. 
Given a filtered probability space $(\Omega,\mathcal F, (\mathcal F_t)_{t\geq 0}, P)$, let the stock price be modeled by $S_t=e^{X_t}$ where $X$ is a Lévy process with characteristic triple $(a^2,\nu,\gamma)$. The assumption that $S_0=1$ is with no loss of generality. Since we study the $L^2$ hedging error, we will always suppose that $S$ is square integrable. 
 The characteristic function of $X$ is denoted by $\phi_t$ and the characteristic exponent by $\psi$: $E[e^{iuX_t}] \equiv \phi_t(u)\equiv e^{t\psi(u)}$.  The process $S$ can be written in the form
\begin{align*}
S_t = 1  + \int_0^t b S_u du  + \int_0^t a S_u dW_u + \int_0^t S_{u-}\int_{\mathbb R}(e^z-1)\tilde J(du\times dz) \, ,
\end{align*}
where $W$ is a standard Brownian motion, $\tilde J$ a compensated Poisson random measure with intensity measure $dt \times \nu$ and $b:= \gamma + \frac{1}{2}a^2 + \int_{\rr} (e^z-1-z 1_{|z| \leq 1}) \nu(\dd z)$. Furthermore, we denote $A:= a^2 + \int_{\mathbb R}(e^z-1)^2 \nu(dz)$. 

We assume that there exists a risk-neutral probability $Q\sim P$, such that the prices of all assets are martingales under $Q$ (the interest rate is assumed to be zero). Moreover, we assume that $X$ is a Lévy process under $Q$ with characteristic exponent $\bar\psi$, characteristic function $\bar\phi_t$ and Lévy measure $\bar\nu$. 

\paragraph{Option pricing}
Consider a European option with pay-off $G(S_T)$ at time $T$ and denote by $g$ its log-payoff function: $G(e^x)\equiv g(x)$. Prices of European options can be computed from the risk-neutral characteristic function $\bar\phi$.
\begin{proposition}\label{optionfourier.prop}${}$
\begin{itemize}
\item[(i)]Suppose that there exists $R\in \mathbb R$ such that
\begin{align}
&g(x)e^{-Rx} \quad \text{has finite variation on $\mathbb R$,}\label{condopt1}\\
&g(x)e^{-Rx}\in L^1(\mathbb R),\label{condopt3}\\
&E^Q[e^{RX_{T-t}}]<\infty \quad \text{and}\quad \int_{\mathbb R}\frac{|\bar\phi_{T-t}(u-iR)|}{1+|u|}du < \infty.\label{condopt4}
\end{align}
Then
\begin{align}
C(t,S_t):=E^Q[G(S_T)|\mathcal F_t] = \frac{1}{2\pi}\int_{\mathbb R+iR}\hat g(u)\bar\phi_{T-t}(-u)S_t^{-iu}du,\label{optionfourier}
\end{align}
where 
$$
\hat g(u) := \int_{\mathbb R} e^{iux}g(x)dx
$$
and moreover
\begin{align}
|\hat g(u+iR)|\leq \frac{C}{1+|u|},\quad u\in \mathbb R\label{digibound}
\end{align}
for some $C>0$.
\item[(ii)]
Suppose that $g$ is differentiable and there exists $R\in \mathbb R$ such that
\begin{align}
&g'(x)e^{-Rx} \quad \text{has finite variation on $\mathbb R$,}\label{europt1}\\
&g(x)e^{-Rx}\in L^1(\mathbb R)\quad \text{and}\quad g'(x)e^{-Rx}\in L^1(\mathbb R)\label{europt3}\\
&E^Q[e^{RX_{T-t}}]<\infty.\label{europt4}
\end{align}
Then the representation \eqref{optionfourier} holds and 
\begin{align}
|\hat g(u+iR)|\leq \frac{C}{1+|u|^2},\quad u\in \mathbb R\label{eurobound}
\end{align}
for some $C>0$.
\end{itemize}
\end{proposition}
For the proof, see \cite{tankov.09}.

For the digital option with pay-off $G(S_T) = 1_{S_T \geq K}$, conditions \eqref{condopt1} and \eqref{condopt3} are satisfied for all $R>0$ and
$$
\hat g(u+iR) = \frac{K^{iu-R}}{R-iu}.
$$
For the European call option with pay-off $G(S_T) = (S_T-K)^+$, conditions \eqref{europt1} and \eqref{europt3} are satisfied for all $R>1$ and
$$
\hat g(u+iR) = \frac{K^{iu+1-R}}{(R-iu)(R-1-iu)}.
$$ 
In any case, conditions \eqref{condopt1} and \eqref{condopt3} imply 
$|G(S)|\leq C S^R$ for some $C>0$ and all $S>0$.

In this paper, we study the behavior of the discretization error for the commonly used hedging strategies: the delta hedging strategy and the quadratic hedging strategy.
Our method is based on the integral representation for the strategy $F$ of the form
\begin{align}
F_t = F_0 + \int_0^t \mu_u du + \int_0^t \sigma_u dW_u + \int_0^t \int_{\mathbb R}\gamma_{u-}(z)\tilde J(du\times dz),\quad \forall t<T.\label{levyito}
\end{align}
Below we show how this representation can be obtained for the strategies we are interested in.

\paragraph{Delta hedging}
The delta hedging strategy is the classical hedging strategy inherited from the Black-Scholes model and given by $F_t = \frac{\partial C(t,S_t)}{\partial S}$. It is not optimal in exponential Lévy models but is nevertheless commonly used by market practitioners. 
\begin{proposition}[Delta hedging]\label{delta.prop}
Let the conditions \eqref{condopt1}, \eqref{condopt3} and \eqref{condopt4} for all $t<T$ be satisfied, and assume that
\begin{align}
\int_{\mathbb R+iR}|\hat g(u)\bar\phi_{T-t}(-u)(-iu)|du<\infty,\quad \forall t<T. \label{deltacond.eq}
\end{align}
Then the delta hedging strategy is given by
\begin{align}
F_t = \frac{\partial C(t,S_t)}{\partial S}= \frac{1}{2\pi}\int_{\mathbb R+iR}\hat g(u)\bar\phi_{T-t}(-u)(-iu)S_t^{-iu-1}du.\label{delta.eq}
\end{align}
Assume in addition
\begin{align}
&\int_{|x|>1}e^{2(R-1)x}\nu(dx)<\infty \, .
\end{align}
Then the representation \eqref{levyito} holds for $F$ with \comment{Added `,' and `.' to the equations below.}
\begin{align}
\mu_t &= \frac{1}{2\pi} \int_{\mathbb R+iR}\hat g(u)\bar\phi_{T-t}(-u)(-iu)S_t^{-1-iu}(\psi(u+i)-\bar\psi(-u))du\label{deltarep1} \, , \\
\sigma_t &= \frac{a}{2\pi} \int_{\mathbb R+iR} \hat g(u)\bar\phi_{T-t}(-u)(-iu)(-1-iu)S_t^{-1-iu}du \, , \\
\gamma_t(z)&= \frac{1}{2\pi} \int_{\mathbb R+iR} \hat g(u)\bar\phi_{T-t}(-u)(-iu) S_t^{-1-iu}(e^{(-1-iu)z}-1)du \, . \label{deltarep3}
\end{align}
\end{proposition}
\begin{proof}
The expression \eqref{delta.eq} is deduced directly from \eqref{optionfourier} using the dominated convergence theorem and the condition \eqref{deltacond.eq}. The martingale representation follows by applying Lemma \ref{martrep.prop} with $f(u):=\frac{1}{2\pi}\hat g(u+iR)(R-iu)$ and $R'=R-1$.
\end{proof}

\paragraph{Quadratic hedging under the martingale probability}
Quadratic hedging in the literature comes in three different flavors: one can (i) minimize the global $L^2$ hedging error computed under the martingale probability (as in \cite{follmer.sondermann.86} and many subsequent papers); (ii) minimize the local variation of the hedging portfolio under the historical probability (as in e.g. \cite{follmer.schweizer.91}) or (iii) minimize the global $L^2$ hedging error under the historical probability (as in \cite{kallsen.hubalek.al.06,cerny.kallsen.07}). 
In this paper we choose the martingale approach, that is, we minimize
\begin{align}
E^Q\left[\left(G(S_T)-C(0,S_0)-\int_0^T F_{t}dS_t\right)^2\right]. \label{minrisk.eq}
\end{align}
This approach is the simplest of the three and thus enables us to explain the main ideas and insights in a less technical setting. See  \cite{cont.al.05} for some arguments towards using this strategy in practice rather than minimizing the quadratic hedging error under the historical measure. Our methodology can also be applied to the local risk minimization, and yields the same results, with more technicalities and under appropriately modified assumptions. 

The solution to the minimization problem \eqref{minrisk.eq} is given by the Kunita-Watanabe decomposition, and can be explicitly written as (see \cite{cerny.kallsen.07})
$$
F_t = \frac{d\langle C,S \rangle^Q _t}{d \langle S, S\rangle_t^Q},
$$
where we denote $C_t:=C(t,S_t)$. 
\begin{proposition}\label{martquad.prop}
Assume \eqref{condopt1}--\eqref{condopt3}; \eqref{condopt4} for all $t<T$ and
\begin{align}
\int_{|x|>1}e^{2(Rx\vee x)}\bar\nu(dx)<\infty.\label{squareint.eq}
\end{align}
Then the quadratic hedging strategy under the martingale probability is given by 
\begin{align}
&F_t = \frac{1}{2\pi}\int_{\mathbb R+iR}\hat g(u)\bar\phi_{T-t}(-u)S_t^{-iu-1}\Upsilon(u)du \label{martquad.eq}\\
&\text{where}\quad \Upsilon(u) = \frac{\bar\psi(-u-i))-\bar\psi(-u)-\bar\psi(-i)}{\bar\psi(-2i)-2\bar\psi(-i)}.
\end{align}
Assume in addition that
$$
\int_{|x|>1}e^{2(R-1)x}\nu(dx)<\infty.
$$
Then the representation \eqref{levyito} holds for $F$ with \comment{Added `,' and `.' to the equations below.}
\begin{align}
\mu_t &= \frac{1}{2\pi} \int_{\mathbb R+iR}\hat g(u)\bar\phi_{T-t}(-u)\Upsilon(u)S_t^{-1-iu}(\psi(-u+i)-\bar\psi(-u))du \, , \label{muquad.eq}\\
\sigma_t &= \frac{a}{2\pi} \int_{\mathbb R+iR} \hat g(u)\bar\phi_{T-t}(-u)\Upsilon(u)(-1-iu)S_t^{-1-iu}du \, , \\
\gamma_t(z)&= \frac{1}{2\pi} \int_{\mathbb R+iR} \hat g(u)\bar\phi_{T-t}(-u)\Upsilon(u) S_t^{-1-iu}(e^{(-1-iu)z}-1)du \, . \label{gammaquad.eq}
\end{align}
\end{proposition}
\begin{proof}
The first part of this result (expression \eqref{martquad.eq} for the optimal strategy) is proved in \cite[Proposition 7]{tankov.09}; under slightly different conditions this result also follows from the general theorem in \cite{kallsen.hubalek.al.06}.
 
To obtain the martingale representation \eqref{muquad.eq}--\eqref{gammaquad.eq}, we apply, once again, Lemma \ref{martrep.prop}, with $R' = R-1$ and $f(u) = \frac{1}{2\pi} \hat g(u+iR)\Upsilon(u+iR)$. The validity of condition \eqref{martcond.eq} follows from Lemma \ref{lem:UVpsi}, assumption \eqref{condopt4} and assumption \eqref{squareint.eq}.
\end{proof}

\section{Errors from discrete hedging: general result}
\label{general.sec}
Since continuously rebalancing one's portfolio is unfeasible in practice, we assume that the hedging portfolio is rebalanced at equally spaced dates $T_i = iT/n$, $i=0,\dots,n-1$, and denote by $h$ the distance between the rebalancing dates: $h:=T/n$. For $t\in(0,T]$ we denote by $\underline{\eta}(t)$ the rebalancing date immediately before $t$ and by $\overline{\eta}(t)$ the rebalancing date immediately after $t$:
$$
\underline{\eta}(t) = \sup \{T_i,T_i< t\},\qquad \overline{\eta}(t) = \inf \{T_i,T_i\geq t\}.
$$
The trading strategy is therefore piecewise constant and is assumed to be given by
$F_{\underline{\eta}(t)}$, where $(F_t)$ is the `ideal' continuous-time hedging strategy that the agent would use if continuous rebalancing were possible. The value of the hedging portfolio at time $t$ is $V_0+\int_0^t F_{s-} dS_s$ with continuous hedging and $V_0+\int_0^t F_{\eta(s)} dS_s$ with discrete hedging.
$F^h_t$ denotes the left-continuous difference between the continuously rebalanced strategy and the discretely rebalanced one: $F^h_t:=F_{t-} - F_{\underline{\eta}(t)}$. We study the $L^2$ convergence to $0$, when $h \to 0$, of the difference between discrete and continuous hedging portfolio
$$
\int_0^T (F_{t-} -F_{\underline{\eta}(t)})dS_t\equiv \int_0^T F^h_{t} dS_t. 
$$

Choose a function $r(h):(0,\infty)\to(0,\infty)$ with $\lim_{h\downarrow 0} r(h) = 0$ (the rate of convergence to zero of the hedging error). We shall see that under suitable assumptions $\ee [(\int_0^T F^h_{t} dS_t)^2/r(h)]$ converges to a finite nonzero limit when $h \downarrow 0$.

\begin{theorem}\label{levyerror}
Assume that the hedging strategy $F$ is of the form \eqref{levyito} and \comment{Added `.' and `,' in the equations below.}
\begin{align}
\lim_{h\downarrow 0} \frac{h}{r(h)} E\left[\int_0^T S_t^2(\overline{\eta}(t)-t)\left(\mu_t^2 + \int_{\mathbb R} \gamma_t^2(z) \nu(dz)\right)dt\right]=0 \, . \label{ass1}
\end{align}
Then
\begin{align}
\lim_{h\downarrow 0}\frac{1}{r(h)}E\left[\left(\int_0^T F^h_{t}dS_t\right)^2\right] = \lim_{h\downarrow 0} \frac{A}{r(h)} E\left[\int_0^T S_t^2(\overline{\eta}(t)-t)\left(\sigma_t^2 + \int_{\mathbb R} \gamma_t^2(z) e^{2z}\nu(dz)\right)dt\right] \, , \label{limit} 
\end{align}
whenever the limit on the right-hand side exists.
\end{theorem}
\begin{corollary}\label{cor:levyerror}
Assume that \eqref{ass1} is satisfied and
\begin{align}
 E\left[\int_0^T S_t^2\left(\sigma_t^2 + \int_{\mathbb R} \gamma_t^2(z) e^{2z}\nu(dz)\right)dt\right]<\infty.  \label{cor:ass1}
\end{align} 
Then
\begin{align*}
\lim_{h\downarrow 0}\frac{1}{h}E\left[\left(\int_0^T F^h_{t}dS_t\right)^2\right] =  \frac{A}{2} E\left[\int_0^T S_t^2\left(\sigma_t^2 + \int_{\mathbb R} \gamma_t^2(z) e^{2z}\nu(dz)\right)dt\right] \, .
\end{align*} 
\end{corollary}
If the condition \eqref{cor:ass1} is not satisfied, then clearly the limit in \eqref{limit} can only exist with a convergence rate worse than $r(h)=h$. 
Therefore the best possible convergence rate which can be obtained with Theorem \ref{levyerror}, and which is realized for regular strategies, is $r(h)=h$. However, worse rates may arise in the presence of irregular pay-offs. In the following, we will refer to the situation when \eqref{cor:ass1} is satisfied and $r(h)=h$ as \emph{regular regime} and to the other situations as \emph{irregular regime}. 

\begin{proof}[Proof of Corollary \ref{cor:levyerror}]
The proof is very similar to that of the Riemann-Lebesgue lemma. Let 
$$
g_h(t):=\frac{\bar\eta(t)-t}{h},\quad f(t):= E\left[ S_t^2\left(\sigma_t^2 + \int_{\mathbb R} \gamma_t^2(z) e^{2z}\nu(dz)\right)\right].
$$
For any piecewise constant function $u:[0,T]\to \mathbb R$, we clearly have
$$
\lim_{h\downarrow 0} \int_0^T g_h(t) u(t)dt = \frac{1}{2}\int_0^T u(t)dt. 
$$
Let $(f_n)_{n\geq 1}$ be a sequence of piecewise constant functions satisfying $f_n(t)\leq f_{n+1}(t) \leq f(t)$ and $\lim_{n\to \infty} f_n(t) = f(t)$ for all $t\in [0,T]$. Then, by monotone convergence, since $|g_h|\leq 1$,
$$
\lim_{n\to \infty} \int_0^T (f(t)-f_n(t))g_h(t)dt = 0,
$$
uniformly on $h$,
which proves that
$$
\lim_{h \downarrow 0} \int_0^T f(t)g_h(t)dt  =\frac{1}{2}\int_0^T f(t)dt.
$$
\end{proof}

\begin{proof}[Proof of Theorem \ref{levyerror}]
We define auxiliary probability measures $P^1$ and $P^2$ by
$$
\frac{dP^k}{dP}\Big|_{\mathcal F_t} := \frac{e^{kX_t}}{e^{t\psi(-ik)}},\quad k=1,2.
$$
Under $P^k$, the process $(W^{(k)}_t)$ defined by $W^{(k)}_t = W_t - akt$ is a standard Brownian motion and $$\tilde J^{(k)}(dt\times dz) = \tilde J(dt\times dz) - dt \times (e^{kz}-1)\nu(dz)$$ is a compensated Poisson random measure. Therefore, the drift of $F$ under $P^k$ is given, by
\begin{align}
\mu^{(k)}_t &= \mu_t + ak\sigma_t + \int_{\mathbb R}\gamma_t(z)(e^{kz}-1)\nu(dz).\label{mu1}
\end{align} 

The hedging error satisfies
\begin{align}
&\frac{1}{r(h)}E\left[\left(\int_0^T F^h_{t}dS_t\right)^2\right] = \frac{1}{r(h)}E\left[\left(\int_0^T F^h_{t}dS^m_t\right)^2\right] \notag\\ &\qquad+\frac{b^2}{r(h)}E\left[\left(\int_0^T F^h_{t}S_t dt\right)^2\right]+ \frac{2b}{r(h)}E\left[\int_0^T F^h_{t}S_t dt\times \int_0^T F^h_{t}dS^m_t\right],\label{twoterms} 
\end{align}
where $S^m$ denotes the martingale part of $S$. The first term in the right-hand side satisfies
$$
E\left(\int_0^T F^h_{t}dS^m_t\right)^2 = A E\left[\int_0^T (F^h_{t})^2 S_t^2 dt\right] = A\int_0^T e^{t\psi(-2i)} E^{P^2}[(F^h_{t})^2] dt
$$
if the expectations are finite. The expectation under the integral sign can be decomposed as follows:
\begin{align}
&E^{P^2}[(F^h_t)^2] = E^{P^2}\left[\left(\int_{\underline{\eta}(t)}^t \mu_s^{(2)}ds\right)^2\right] \notag\\&\qquad + E^{P^2}\left[\left(\int_{\underline{\eta}(t)}^t \sigma_s dW^{(2)}_s + \int_{\underline{\eta}(t)}^t\int_{\mathbb R} \gamma_{s-}(z)\tilde J^{(2)}(ds\times dz)\right)^2\right]\notag \\ &\qquad + E^{P^2}\left[\left(\int_{\underline{\eta}(t)}^t \sigma_s dW^{(2)}_s + \int_{\underline{\eta}(t)}^t\int_{\mathbb R} \gamma_{s-}(z)\tilde J^{(2)}(ds\times dz)\right)\int_{\underline{\eta}(t)}^t \mu_s^{(2)}ds\right].\label{underint}
\end{align}
The second term in the right-hand side above satisfies
\begin{align*}
E^{P^2}\left[\left(\int_{\underline{\eta}(t)}^t \sigma_s dW^{(2)}_s + \int_{\underline{\eta}(t)}^t\int_{\mathbb R} \gamma_{s-}(z)\tilde J^{(2)}(ds\times dz)\right)^2\right] = E^{P^2}\left[\int_{\underline{\eta}(t)}^t\left(\sigma_s^2 + \int_{\mathbb R}\gamma^2_s(z)e^{2z}\nu(dz)\right) \right]
\end{align*}
and its integral gives, using integration by parts and switching back to the probability $P$,
\begin{align*}
&\frac{A}{r(h)}\int_0^T e^{t\psi(-2i)} E^{P^2}\left[\int_{\underline{\eta}(t)}^t\left(\sigma_s^2 + \int_{\mathbb R}\gamma^2_s(z)e^{2z}\nu(dz)\right)\right]dt\\
&\qquad = \frac{A}{r(h)} \int_0^T dt E^{P^2}\left[\sigma_t^2 + \int_{\mathbb R}\gamma^2_t(z)e^{2z}\nu(dz)\right]\int_t^{\overline{\eta}(t)}e^{s\psi(-2i)}ds\\
&= \frac{A(1+O(h))}{r(h)} E\left[\int_0^T S_t^2(\overline{\eta}(t)-t)\left(\sigma_t^2 + \int_{\mathbb R} \gamma_t^2(z) e^{2z}\nu(dz)\right)\right],
\end{align*}
which converges to the same limit as \eqref{limit}. In view of this result and of the fact that the cross terms in \eqref{twoterms} and \eqref{underint} can be estimated using the Cauchy-Schwartz inequality, to prove the theorem it remains to show that, under the assumptions and when the limit in the right-hand side of \eqref{limit} exists, 
\begin{align}
&\lim_{h\downarrow 0} \frac{1}{r(h)} \int_0^T e^{t\psi(-2i)}E^{P^2}\left[\left(\int_{\underline{\eta}(t)}^t \mu_s^{(2)}ds\right)^2\right]dt = 0\label{driftstrat}\\ 
\text{and}\quad  &\lim_{h\downarrow 0} \frac{1}{r(h)}E\left[\left(\int_0^T F^h_{t}S_t dt\right)^2\right]=0.\label{driftasset}
\end{align}
\paragraph{Proof of \eqref{driftstrat}}
The expression under the $\lim$ sign satisfies
\begin{align*}
&\frac{1}{r(h)} \int_0^T e^{t\psi(-2i)}E^{P^2}\left(\int_{\underline{\eta}(t)}^t \mu_s^{(2)}ds\right)^2dt\leq \frac{h}{r(h)} \int_0^T e^{t\psi(-2i)}E^{P^2}\left[\int_{\underline{\eta}(t)}^t (\mu_s^{(2)})^2ds\right]dt\\
&\qquad = \frac{h(1+O(h))}{r(h)}\int_0^T (\overline{\eta}(t)-t)E\left[S_t^2 (\mu_t^{(2)})^2\right]dt \\
&\qquad \leq \frac{Ch}{r(h)}\int_0^T (\overline{\eta}(t)-t)E\left[S_t^2 \mu_t^2 + S_t^2 \sigma_t^2 + S_t^2 \left(\int_{\mathbb R}\gamma_t(z)(e^{2z}-1)\nu(dz)\right)^2\right]dt,
\end{align*}
for some constant $C<\infty$, where the last estimate follows from \eqref{mu1}. By the Jensen inequality (for $|z|>1$) and the Cauchy-Schwartz inequality (for $|z|\leq 1$), 
\begin{align*}
&\left(\int_{\mathbb R}\gamma_t(z)(e^{2z}-1)\nu(dz)\right)^2 \leq  2\int_{|z|\leq 1}(e^{2z}-1)^2 \nu(dz) \int_{|z|\leq 1}\gamma_t^2(z)\nu(dz)\\ &\qquad+ 2 \int_{|z|>1}|e^{2z}-1|\nu(dz)\int_{|z|>1}\gamma_t^2(z)|e^{2z}-1|\nu(dz)
\leq C \int_{\mathbb R}\gamma_t^2(z)(1+e^{2z})\nu(dz).
\end{align*}

The limit \eqref{driftstrat} now follows from the assumptions of the theorem and the existence of the limit \eqref{limit}. 
\paragraph{Proof of \eqref{driftasset}} 
The error term in \eqref{driftasset} can be rewritten as
\begin{align*}
&\frac{1}{r(h)}E\left[\left(\int_0^T F^h_{t}S_t dt\right)^2\right] = \frac{2}{r(h)}\sum_{i=1}^n \int_{T_{i-1}}^{T_i}dt \int_{t}^{T_i}ds E[(F_t-F_{T_{i-1}})(F_s-F_{T_{i-1}})S_tS_s] \\
&\qquad + \frac{2}{r(h)}\sum_{1\leq i<j \leq n} \int_{T_{i-1}}^{T_i}dt \int_{T_{j-1}}^{T_j}ds E[(F_t-F_{T_{i-1}})(F_s-F_{T_{j-1}})S_tS_s]\\
&= \frac{2}{r(h)}\sum_{i=1}^n \int_{T_{i-1}}^{T_i}dt \int_{t}^{T_i}ds e^{(s-t)\psi(-i)}E[(F_t-F_{T_{i-1}})^2 S_t^2]\\
&\qquad + \frac{2}{r(h)}\sum_{i=1}^n \int_{T_{i-1}}^{T_i}dt \int_{t}^{T_i}ds e^{s\psi(-i)}E^{P^1}[(F_t-F_{T_{i-1}})(F_s-F_t)S_t]\\
&\qquad + \frac{2}{r(h)}\sum_{1\leq i<j \leq n} \int_{T_{i-1}}^{T_i}dt \int_{T_{j-1}}^{T_j}ds e^{s\psi(-i)} E^{P^1}[(F_t-F_{T_{i-1}})(F_s-F_{T_{j-1}})S_t]\\
&= \frac{O(h)}{r(h)}\int_{0}^{T}dt E[(F_t-F_{\underline{\eta}(t)})^2 S_t^2]+\frac{2}{r(h)}\int_0^T dt \int_t^T ds e^{s\psi(-i)}E^{P^1}[(F_t-F_{\underline{\eta}(t)})(F_s-F_{\underline{\eta}(s)\vee t})S_t].
\end{align*}
The first term in the last line converges to zero by the first part of the proof. To compute the second term, we introduce the conditional expectation with respect to $\mathcal F_{\underline{\eta}(s)\vee t}$ inside the expectation $E^{P1}$.
The fact that the local martingale part of $F$ has zero expectation can be justified using \eqref{limit}. Finally, we get
\begin{align*}
&\left|\frac{2}{r(h)}\int_0^T dt \int_t^T ds e^{s\psi(-i)}E^{P^1}[(F_t-F_{\underline{\eta}(t)})(F_s-F_{\underline{\eta}(s)\vee t})S_t]\right|\\
&= \left|\frac{2}{r(h)}\int_0^T dt \int_t^T ds e^{s\psi(-i)}E^{P^1}\left[(F_t-F_{\underline{\eta}(t)})S_t\int_{\underline{\eta}(s)\vee t}^s \mu^{(1)}_u du\right]\right|\\
&= \left|\frac{2}{r(h)}\int_0^T dt \int_t^T ds E[\mu^{(1)}_s (F_t - F_{\eta(t)})S_t S_s]\int_s^{\overline{\eta}(s)}e^{(u-s)\psi(-i)}du\right| \\
&\leq \frac{2}{r(h)} \int_0^T dt E[(F_t - F_{\eta(t)})^2S^2_t]^{\frac{1}{2}}  \int_0^T ds E[(\mu^{(1)}_s)^2 S_s^2]^{\frac{1}{2}} \int_s^{\overline{\eta}(s)}e^{(u-s)\psi(-i)}du\\
&\leq C\left(\frac{1}{r(h)} \int_0^T dt E[(F_t - F_{\eta(t)})^2S^2_t]\right)^{\frac{1}{2}} \left(\frac{1}{r(h)} \int_0^T ds E[(\mu^{(1)}_s)^2 S_s^2] \left\{\int_s^{\overline{\eta}(s)}e^{(u-s)\psi(-i)}du\right\}^2 \right)^{\frac{1}{2}}, 
\end{align*}
where the last implication follows from the Jensen inequality. The first factor in the in the right-hand side above was shown to be bounded in the beginning of the proof. As for the second factor,
\begin{align*}
\frac{1}{r(h)} \int_0^T ds E[(\mu^{(1)}_s)^2 S_s^2] \left\{\int_s^{\overline{\eta}(s)}e^{(u-s)\psi(-i)}du\right\}^2 = \frac{O(h)}{r(h)} \int_0^T ds (\overline{\eta}(s)-s) E[(\mu^{(1)}_s)^2 S_s^2], 
\end{align*}
which can be shown to converge to zero in the same way as we did in the proof of \eqref{driftstrat}.
\end{proof}

\section{Convergence rates of specific strategies}
\label{specific.sec}
We start by introducing a set of assumptions on the Lévy measure $\nu$ of $X$, which will be used in different theorems later in this section. 
In theorems dealing with the delta-hedging strategy we require:
\begin{align}
\int_{|x|>1}e^{Rx}\bar\nu(dx)<\infty,\quad \int_{|x|>1}e^{2(Rx\vee x)}\nu(dx)<\infty,\quad \text{and}\quad \int_{|x|>1}e^{2(R-1)x}\nu(dx)<\infty.\label{deltaint.eq}
\end{align}
The first condition guarantees the integrability of the option payoff under $Q$ (recall that pay-off function satisfies $|G(S)|\leq C S^R$), the second ensures the square integrability of the option price and the stock price under $P$, and the last condition allows to construct a martingale-drift representation for the strategy.  

For analyzing the quadratic hedging under the martingale probability we require the stock price and the option pay-off to be square integrable under $Q$ as well:
\begin{align}
&\int_{|x|>1}e^{2(Rx\vee x)}\bar\nu(dx)<\infty,\quad \int_{|x|>1}e^{2(Rx\vee x)}\nu(dx)<\infty,\quad \text{and}\quad \int_{|x|>1}e^{2(R-1)x}\nu(dx)<\infty.\label{quadint.eq}
\end{align}

The following alternative assumptions determine the decay properties of characteristic function of $X$ at infinity. 
\begin{itemize}
\item[(H1)] The Lévy measure $\nu$ is of the form $\nu = \nu_0 + \nu_1$ where $\nu_0$ is a finite measure on $\mathbb R$ and $\nu_1$ has a positive density of the form
$$
\nu_1(x) = \frac{k(x)}{|x|},
$$ 
where the function $k$ is right-continuous and increasing on $(-\infty,0)$ and left-continuous and decreasing on $(0,\infty)$.
\item[(H2-$\alpha$)] The Lévy measure $\nu$ satisfies
$$
\limsup_{r\to 0} r^{\alpha-2}\int_{[-r,r]}x^2 \nu(dx)>0. 
$$
\label{H1}
\item[(H3-$\alpha$)] The Lévy measure $\nu$ satisfies
$$
\int_{[-1,1]} |x|^{\alpha}\nu(dx)<\infty. 
$$
\item[(H4-$\alpha$)] The Lévy measure $\nu$ has a density satisfying
\begin{align*}
\nu(x) = \frac{f(x)}{|x|^{1+\alpha}}, \quad \lim_{x\to 0+}f(x) = f_+,\quad \lim_{x\to 0-}f(x) = f_-
\end{align*}
for some constants $f_->0$ and $f_+>0$.

\end{itemize}

The assumption H1 guarantees at least power-law decay of the characteristic function at infinity (see Lemma \ref{cond.lm} in the Appendix). It is satisfied by most parametric infinite intensity processes used in financial modeling: for the variance gamma \cite{madan98} and CGMY \cite{finestructure} processes this is immediately clear by looking at the Lévy measure while for the normal inverse Gaussian process \cite{bns_nig} and the generalized hyperbolic distribution \cite{eberlein} it follows from the self-decomposability of these distributions shown in \cite{halgreen} and the characterization of self-decomposable distributions in \cite[Chapter 3]{sato}.

The assumptions H2-$\alpha$, H3-$\alpha$ and H4-$\alpha$ with $0<\alpha<2$ characterize different aspects of stable-like behavior of small jumps of the Lévy process. They are satisfied by the CGMY process (with $\alpha=Y$), the normal inverse Gaussian process (with $\alpha=1$) and the generalized hyperbolic distribution (with $\alpha=1$ in general; see \cite[pages 125--126]{levybook}). They are not satisfied by the variance gamma process. It is clear that the assumption H4-$\alpha$ implies H2-$\alpha$ and H3-$\alpha$.  

We start our analysis with regular pay-offs, in which case the convergence takes place in the regular regime with the rate $r(h)=h$. In the following theorem and its proof, we use the notation of Propositions \ref{delta.prop} and \ref{martquad.prop}.

\begin{theorem}[Regular pay-offs] \label{thm:euro_delta} Let the pay-off function and the Lévy process satisfy the conditions \eqref{europt1}--\eqref{europt3} for some $R\in \mathbb R$ and assume that one of the three alternative conditions holds: 
\begin{itemize}
\item $\nu$ satisfies the assumption H1 and $a=0$; 
\item $\nu$ satisfies the assumption H2-$\alpha$ with $\alpha\in (0,2)$ and $a=0$; 
\item $a>0$. 
\end{itemize}
Let the hedging strategy be given by Proposition \ref{delta.prop} and assume that \eqref{deltaint.eq} holds or let the hedging strategy be given by Proposition \ref{martquad.prop} and assume that \eqref{quadint.eq} holds.
Then
\begin{equation}
\lim_{h \downarrow 0} \frac{1}{h} \ee \left[ \left( \int_0^T F_t^h  \dd S_t \right)^2 \right]  = \frac{A}{2} \ee \Biggl[ \int_0^T S_t^2 \left(\sigma_t^2 +  \int_{\rr}  \gamma_t(z) e^{2z} \nu(\dd z)\right) \dd t \Biggr] \, .\label{result.eq}
\end{equation} 
\end{theorem}

\begin{remark}
As will become clear from the proof, for example, for delta hedging, the limiting renormalized discretization error can be evaluated via a two-dimensional integral.
\begin{align*}
&\lim_{h \downarrow 0} \frac{1}{h} \ee \left[ \left( \int_0^T F_t^h  \dd S_t \right)^2 \right]\\ &\qquad = \frac{A}{8\pi^4} \int_{\mathbb R +iR} \int_{\mathbb R +iR}  du_1\,du_2  (\phi_T(-u_1-u_2)-\bar \phi_{T}(-u_1)\bar \phi_{T}(-u_2)) \hat g(u_1) \hat g(u_2)f(u_1,u_2) , \\ 
&\text{where}\, \, f(u_1,u_2) = -u_1 u_2 \frac{\psi(-u_1-u_2) - \psi(-u_1-i) - \psi(-u_2-i) + \psi(-2i)}{\psi(-u_1-u_2)-\bar \psi(-u_1)-\bar \psi(-u_2)}.
\end{align*}
In any case, our goal in this paper is not to compute the hedging error explicitly but rather to gain an understanding of its behavior as the rebalancing step tends to zero.    
\end{remark}

\begin{proof}
\textit{Step 1.} From Lemma \ref{cond.lm} or, under the condition $a>0$, directly from the form of the characteristic function, and from Lemma \ref{lem:PQUVpsi} it follows that
$$
\int_{\mathbb R} \frac{|\bar \phi_{T-t}(u-iR)|}{1+|u|}du <\infty,\quad \forall t<T,
$$
and therefore Proposition \ref{delta.prop} holds.

 With $\mu$, $\sigma$ and $\gamma$ as in \eqref{deltarep1}--\eqref{deltarep3} for the delta hedging strategy or as in \eqref{muquad.eq}--\eqref{gammaquad.eq} for the quadratic hedging strategy, define 
\begin{align}
& I_1(t) := \ee [S_t^2 \mu_t^2]  \, , && I_2(t) := \ee [S_t^2 \sigma_t^2]  \, , \label{i1i2}\\
& I_3(t) := \ee [S_t^2 \int_{\rr} \gamma_t^2(z) \nu(\dd z)]  \, , && I_4(t) := \ee [S_t^2 \int_{\rr} \gamma_t^2(z) e^{2 z} \nu(\dd z)]  \, .\label{i3i4}
\end{align}
Suppose that we can show that $\int_0^T I_i(t) \dd t < \infty$ for $i \in\{1,2,3,4\}$. Then assumption \eqref{cor:ass1} of Corollary \ref{cor:levyerror} is satisfied and assumption \eqref{ass1} of Theorem \ref{levyerror} is satisfied as well ($r(h)=h$). Therefore, by an application of Corollary \ref{cor:levyerror} the proof is complete.

\textit{Step 2.}
For the delta hedging strategy (proposition \ref{delta.prop}, ) from equations \eqref{deltarep1}--\eqref{deltarep3} and the bound \eqref{eurobound},
\begin{align*}
&I_i(t) \leq  C\int_{\rr+iR} \int_{\rr+iR} \frac{|f_i(u_1,u_2)\bar\phi_{T-t}(-u_1)\bar\phi_{T-t}(-u_2)\phi_t(-u_1-u_2)|}{(|u_1|+1)(|u_2|+1)}du_1 du_2
\end{align*}
for some $C>0$, where
\begin{align*}
f_1(u_1,u_2) &= (\psi(-u_1+i)-\bar\psi(-u_1))(\psi(-u_2+i)-\bar\psi(-u_2)), \\
 f_2(u_1,u_2) &= a^2 (-1-iu_1)(-1-iu_2) , \\
f_3(u_1,u_2) &= \int_{\rr} (e^{(-1-iu_1)z} - 1)(e^{(-1-iu_2)z} - 1)  \nu(\dd z) \\
& = \psi(-u_1-u_2+2i) - \psi(-u_1+i) - \psi(-u_2+i)) - f_2(u_1,u_2)
\intertext{and}
 f_4(u_1,u_2) &= \int_{\rr} e^{2z} (e^{(-1-iu_1)z} - 1)(e^{(-1-iu_2)z} - 1)  \nu(\dd z)\\  
&=  \psi(-u_1-u_2) - \psi(-u_1-i) - \psi(-u_2-i)+ \psi(-2i) - f_2(u_1,u_2)
\end{align*}
From Lemmas \ref{lem:PQUVpsi} and \ref{lem:UVpsi},
\begin{align}
|f_i(u_1+iR,u_2+iR)| \leq C(1+ \sqrt{|\Re \psi(u_1)|}) (1+ \sqrt{|\Re \psi(u_2)|})\label{fi.bound}
\end{align}
for some $C<\infty$ and $i \in \{1,2,3,4 \}$.  Corollary \ref{repsi.bound} and Lemma \ref{lem:PQUVpsi} then imply that $I_i(t) \leq J(t)$, where the function $J$ is defined by
\begin{align}
J(t) & = C\int_{\rr^2} \frac{(1+ \sqrt{|\Re \psi(u_1)}|)(1+ \sqrt{|\Re \psi(u_2)|})}{(1+|u_1|)(1+|u_2|)} e^{c(\Re \psi(u_1+u_2)t + \Re {\psi}(u_1) (T-t) + \Re {\psi}(u_2) (T-t))}  \dd u_1 \dd u_2\label{jbound}
\end{align}
for some constants $C>0$ and $c>0$ (which will later change from line to line). 

For the quadratic hedging strategy (proposition \ref{martquad.prop}), by the same arguments, we get that
\begin{align*}
I_i(t)  & \leq C\int_{\rr^2} \frac{(1+ \sqrt{|\Re \psi(u_1)}|)(1+ \sqrt{|\Re \psi(u_2)|})|\Upsilon(u_1+iR)\Upsilon(u_2+iR)|}{(1+|u_1|^2)|(1+|u_2|^2)} \\
&  \quad \times e^{c(\Re \psi(u_1+u_2)t + \Re {\psi}(u_1) (T-t) + \Re {\psi}(u_2) (T-t))}  \dd u_1 \dd u_2
\end{align*}
From Lemma \ref{lem:UVpsi} we now get that
$$
|\Upsilon(u+iR)|\leq C(1+\sqrt{|\psi(u)|}) \leq C(1+|u|),
$$
which implies $I_i(t)\leq J(t)$.

 It remains to show that $\int_0^T J(t) \dd t < \infty$, and the theorem will be proved.

\textit{Step 3.} Assume first that $\nu$ satisfies H1. The change of variables $u_1+u_2=v_1$ and $u_1-u_2=v_2$ together with \eqref{eqn:lem_H3} and Lemma \ref{lem:growth_imp} yields
\begin{align*}
J(t) \leq & C\int_{\rr^2}  \frac{(1+\sqrt{|\Re \psi ( (v_1+v_2)/2 )|})(1+ \sqrt{|\Re \psi ( (v_1-v_2)/2 )|})}{(1+|v_1+v_2|)(1+|v_1-v_2|)}   e^{c(\Re \psi( v_1)T+\Re {\psi}(v_2) (T-t))}  \dd v_1 \dd v_2 \, .
\end{align*}
Now by \eqref{eqn:lem_H4}
\begin{equation*}
J(t) \leq  C\int_{\rr^2} \frac{1 + |\Re \psi(v_1/2))|+|\Re \psi(v_2/2)|}{(1+|v_1+v_2|)(1+|v_1-v_2|)}  e^{c(\Re \psi(v_1)T+ \Re {\psi}(v_2) (T-t)}  \dd v_1 \dd v_2 \, .
\end{equation*}

\textit{Step 4.} In this last step we consider the integral of $J(t)$ over $[0,T]$. 
\begin{equation*}
\begin{split}
\int_0^T J(t) \dd t  &\leq  C  \int_{\rr^2} \frac{(1 + |\Re \psi(v_1/2))|+|\Re \psi(v_2/2)|)e^{c \Re \psi(v_1)T}}{(1+|v_1+v_2|)(1+|v_1-v_2|)}  \frac{1-e^{c \Re \psi(v_2) T}}{\Re \psi(v_2)}  \dd v_1 \dd v_2  \\
& \leq  C  \int_{\rr^2} \frac{(1+|\Re \psi(v_1/2))|+|\Re \psi(v_2/2)|)e^{c\Re \psi(v_1)T}}{(1+|v_1+v_2|)(1+|v_1-v_2|)} \frac{1}{1 + |\Re \psi(v_2)|} \dd v_1 \dd v_2 \\
& \leq C  \int_{\rr^2}  \frac{(1+|\Re \psi(v_1)|) e^{c \Re \psi(v_1)T} }{(1+|v_1+v_2|)(1+|v_1-v_2|)}  \dd v_1 \dd v_2 ,
\end{split}
\end{equation*}
where the last inequality follows from \eqref{growth.assum}. 

From Lemma \ref{doubleint} we then get
\begin{equation*}
\int_0^T J(t) \dd t \leq C  \int_{\rr} \frac{(1+|\Re \psi(v_1)|)(1+\log (1 + |v_1|)) e^{c \Re \psi(v_1)} }{1+|v_1|}  \dd v_1,
\end{equation*}
and also
\begin{equation*}
\int_0^T J(t) \dd t \leq C  \int_{\rr} \frac{(1+\log (1 + |v_1|)) e^{c \Re \psi(v_1)} }{1+|v_1|}  \dd v_1,
\end{equation*}
for different constants $c$ and $C$. Lemma now \ref{cond.lm} allows to conclude that this integral is finite, completing the proof of the theorem under the assumption H1.  

Suppose now that one of the two alternative assumptions is satisfied. Then, from Lemma \ref{cond.lm}, or, if $a>0$, directly from the form of the characteristic function, we get
\begin{equation*}
J(t) \leq  C\int_{\rr^2} \frac{(1+|v_1+v_2|^{\alpha/2})(1+|v_1-v_2|^{\alpha/2})}{(1+|v_1+v_2|)(1+|v_1-v_2|)} e^{- cT|v_1|^{\alpha}-c(T-t)|v_2|^\alpha}  \dd v_1 \dd v_2,
\end{equation*}
where we set $\alpha=2$ if $a>0$. To finish the proof in this case, it is now sufficient to repeat the arguments from the beginning of step 3 onwards, taking $\psi(u) = -|u|^\alpha$.
\end{proof}

Next, we turn to options with irregular pay-offs. In this case the convergence rate of the discretization error to zero is not necessarily $r(h)=h$, but depends on the properties of the Lévy measure of $X$ near zero. Therefore, we need to make a precise assumption about these properties. For the same reason (to compute the precise convergence rate and the constant rather than just an upper bound) it is necessary to fix the pay-off profile.  

\begin{theorem}[Delta hedging, digital options] \label{thm:digit_delta} Let the pay-off function be given by $G(S_T) = 1_{S_T\geq K}$ and assume \eqref{deltaint.eq} for some $R>0$.
Let the hedging strategy be given by Proposition \ref{delta.prop}.
\begin{enumerate}
\item Assume that $a=0$ and $\nu$ satisfies the assumption H4-$\alpha$ with $\alpha\in (1,2)$.
Then the hedging error satisfies  
$$
\lim_{h\downarrow 0}\frac{1}{r(h)}E\left[\left(\int_0^T F^h_{t}dS_t\right)^2\right] = \frac{A D_\alpha}{2\pi(f_+ + f_-)^{1/\alpha}} p_T(\log K),
$$
with $r(h) = h^{1-1/\alpha}$, where $D_\alpha$ is a constant depending only on $\alpha$ and given explicitly by
$$
D_\alpha := \frac{1}{(2\Gamma(-\alpha)\cos(\pi(2-\alpha)/2))^{1/\alpha}}\int_{\mathbb R}dv \frac{1-e^{-|v|^\alpha} - |v|^{\alpha}e^{-|v|^\alpha}}{|v|^\alpha (1-e^{-|v|^\alpha})}
$$
and $p_T$ is the density of $X_T$, which can be computed from the characteristic function via
$$
p_T(\log K) = \frac{1}{2\pi} \int_{\mathbb R} dv e^{-iv \log K} e^{T\psi(v)}.
$$

\item Assume that $a>0$. Then the hedging error satisfies  
$$
\lim_{h\downarrow 0}\frac{1}{\sqrt{h}}E\left[\left(\int_0^T F^h_{t}dS_t\right)^2\right] = \frac{A D}{2\pi a} p_T(\log K),
$$
with
\begin{align}
D := \int_{\mathbb R}dv \frac{1-e^{-v^2} - v^{2}e^{-v^2}}{v^2 (1-e^{-v^2})}.\label{D.eq}
\end{align}
\end{enumerate}
\end{theorem}
\begin{proof}
We use the notation introduced in the proof of Theorem \ref{thm:euro_delta}. The proof below covers both cases by setting $\alpha=2$ in the case $a>0$.  
As a preliminary remark, observe that by Lemma \ref{lem:PQUVpsi}, the risk-neutral characteristic exponent $\bar \psi$ also has the property \eqref{powerbelow} of Lemma \ref{cond.lm}. This shows that condition \eqref{deltacond.eq} is satisfied and Proposition \ref{delta.prop} holds. 

\textit{Step 1.} Let
\begin{multline*}
e_i(u_1,u_2,t) := f_i(u_1+iR,u_2+iR)\bar\phi_{T-t}(-u_1-iR)\bar\phi_{T-t}(-u_2-iR) \\
\times \phi_t(-u_1-u_2-2iR)K^{iu_1 + iu_2 - 2R}.
\end{multline*}
Then, with a change of variables, 
\begin{align*}
I_i(t) = \frac{1}{4\pi^2} \int_{\mathbb R^2} e_i(u_1,u_2,t)du_1 du_2 = \frac{h^{-1/\alpha}}{8\pi^2} \int_{\mathbb R^2} e_i\left(\frac{v_1+v_2 h^{-1/\alpha}}{2},\frac{v_1-v_2 h^{-1/\alpha}}{2},t\right)dv_1 dv_2 \, .
\end{align*}
In this first step we would like to show that 
$$
\int_0^T h^{-1}(\bar \eta(t)-t)e_i\left(\frac{v_1+v_2 h^{-1/\alpha}}{2},\frac{v_1-v_2 h^{-1/\alpha}}{2},t\right) dt \, ,
$$
is bounded from above by a function which does not depend on $h$ and is integrable with respect to $v_1$ and $v_2$. This will on one hand prove the assumption \eqref{ass1} (with $r(h) = h^{1-1/\alpha}$) and on the other hand will enable us to use the dominated convergence theorem for computing the limit in \eqref{limit}. 

Using property \eqref{powerbelow}, corollary \ref{repsi.bound}, and the estimate \eqref{fi.bound}, 
$$
|e_i(u_1,u_2,t)| \leq C e^{-c\{t|u_1+u_2|^\alpha +(T-t)|u_1|^\alpha + (T-t)|u_2|^\alpha \}}(1+ |u_1|^{\alpha/2})(1 + |u_2|^{\alpha/2})
$$
for some constants $c,C>0$ which may change from line to line. By Lemma \ref{lem:growth_imp}, we then get:
$$
\left|e_i\left(\frac{v_1+v_2}{2},\frac{v_1-v_2}{2},t\right)\right| \leq C e^{-c\{T|v_1|^\alpha +(T-t)|v_2|^\alpha\}}(1+|v_1|^\alpha +|v_2|^\alpha),
$$
and finally, evaluating the time integral explicitly using the formula
\begin{align}
\int_0^T (\bar \eta(t)-t)e^{a(T-t)}dt = \sum_{i=1}^n e^{a(T-T_i)}\int_{0}^{h} te^{at}dt = \frac{(ah e^{ah} - e^{ah}+1)(1-e^{aT})}{a^2(1-e^{ah})},\label{explicit}
\end{align}
we get
\begin{align*}
&\int_0^T h^{-1}(\bar \eta(t)-t)\left|e_i\left(\frac{v_1+v_2 h^{-1/\alpha}}{2},\frac{v_1-v_2 h^{-1/\alpha}}{2},t\right)\right| dt\\
&\leq Ch^{-2} (h(1+|v_1|^\alpha) + |v_2|^\alpha) e^{-cT|v_1|^\alpha } \int_0^T (\bar \eta(t)-t) e^{-c(T-t)h^{-1}|v_2 |^\alpha}dt\\
& \leq C(1+|v_1|^\alpha) e^{-cT|v_1|^\alpha } 
\frac{-c|v_2|^\alpha e^{-c |v_2|^\alpha} - e^{-c |v_2|^\alpha}+1}{|v_2|^{\alpha}(1 - e^{-c |v_2|^{\alpha}})},
\end{align*}
where the last inequality follows from the bound $1-e^{-x}\leq x$, $x\geq 0$. Since the last expression is integrable with respect to $v_1$ and $v_2$ (it is bounded near zero and behaves like $\frac{1}{|v_2|^\alpha}$ at infinity), step 1 is completed. 

\textit{Step 2.}  Let us now compute the renormalized limiting hedging error \comment{Added `.' in the equation below.}
\begin{align*}
&\varepsilon_0:=\lim_{h\to 0} A h^{1/\alpha-1}E \int_0^T S_t^2 (\bar\eta(t)-t)\left(\sigma_t^2 + \int_{\mathbb R}\gamma_t^2(z)e^{2z}\nu(dz)\right)dt\\
&= \lim_{h\to 0} \frac{A}{8\pi^2}\int_{\mathbb R^2}dv_1 \, dv_2 \int_0^T h^{-1}(\bar \eta(t)-t)\{e_2+e_4\}\left(\frac{v_1+v_2 h^{-1/\alpha}}{2},\frac{v_1-v_2 h^{-1/\alpha}}{2},t\right) dt \, .
\end{align*}
By the dominated convergence theorem, whose application is justified by Step 1, we can compute the limit inside the integral with respect to $v_1$ and $v_2$. 
\begin{align*}
&\lim_{h\to 0} \int_0^T h^{-1}(\bar \eta(t)-t)\{e_2+e_4\}\left(\frac{v_1+v_2 h^{-1/\alpha}}{2},\frac{v_1-v_2 h^{-1/\alpha}}{2},t\right) dt = e^{T\psi(-v_1 - 2iR)}K^{iv_1 - 2R} L_1 L_2,
\end{align*}
with
\begin{align*}
L_1 &= \lim_{h\to 0}h\Biggl\{\psi(-v_1 -2iR) - \psi\left(-\frac{v_1+v_2h^{-1/\alpha}}{2}-i(R+1)\right) \\&\qquad - \psi\left(-\frac{v_1-v_2h^{-1/\alpha}}{2}-i(R+1)\right) +\psi(-2i)\Biggr\},\\
L_2 &= \lim_{h\to 0} \int_0^T dt(\bar \eta(t)-t)h^{-2} e^{(T-t)\{-\psi(-v_1 -2iR) + \bar \psi(-\frac{v_1+v_2h^{-1/\alpha}}{2}-iR) + \bar\psi(-\frac{v_1-v_2h^{-1/\alpha}}{2}-iR)\}}, 
\end{align*}
provided that both limits exist.
Now, a direct computation using Lemma \ref{lem:PQUVpsi} and equations \eqref{psilim} of Lemma \ref{cond.lm} yields $L_1 = 2^{-\alpha}(c_+ + c_-)|v_2|^\alpha$, $v_2 \neq 0$, where the constants $c_+$ and $c_-$ are defined in Lemma \ref{cond.lm}. To compute $L_2$, we first observe that for all $v_2 \neq 0$
\begin{align*}
&\lim_{h\to 0} \Biggl\{-\psi(-v_1 -2iR) + \bar \psi\left(-\frac{v_1+v_2h^{-1/\alpha}}{2}-iR\right) + \bar\psi\left(-\frac{v_1-v_2h^{-1/\alpha}}{2}-iR\right)\Biggr\} = -\infty,\\
& \lim_{h\to 0}h\Biggl\{-\psi(-v_1 -2iR) + \bar \psi\left(-\frac{v_1+v_2h^{-1/\alpha}}{2}-iR\right) \\&\hspace*{3cm}+ \bar\psi\left(-\frac{v_1-v_2h^{-1/\alpha}}{2}-iR\right)\Biggr\} = -2^{-\alpha}(c_++c_-)|v_2|^\alpha \neq 0.
\end{align*}
Combined with the explicit formula \eqref{explicit}, these two limits allow to conclude that
$$
L_2 = \frac{\kappa(v_2)e^{\kappa(v_2)}-e^{\kappa(v_2)}+1}{\kappa(v_2)^2(1-e^{\kappa(v_2)})}, \quad \kappa(v_2) = -2^{-\alpha}(c_++c_-) |v_2|^\alpha.
$$
Finally, assembling $L_1$ and $L_2$ together and performing the integration with respect to $v_1$ and $v_2$, the proof is completed.
\end{proof} 

The behavior of the quadratic hedging strategy for options with irregular pay-off is very different from that of delta hedging: the convergence rate improves rather than deteriorates when the Blumenthal-Getoor index $\alpha$ decreases, and in many cases the convergence takes place in the regular regime even for digital options. 

\begin{theorem}[Martingale quadratic hedging, digital options, regular regime] \label{thm:digit_quad} 
Let the pay-off function and the Lévy process satisfy the conditions \eqref{condopt1}--\eqref{condopt3} and \eqref{quadint.eq} for some $R\in \mathbb R$, and assume that one of the two alternative conditions holds:
\begin{itemize}
\item $a=0$ and $\nu$ satisfies the assumptions H1 and H3-$\alpha_+$ for some $\alpha_+\in(0,1]$.
\item $a=0$ and $\nu$ satisfies the assumptions H2-$\alpha_-$ and H3-$\alpha_+$ with $0<\alpha_- \leq \alpha_+ < \frac{3}{2}$.
\end{itemize}
Let the hedging strategy be given by Proposition \ref{martquad.prop}.
Then
\begin{equation*}
\begin{split}
\lim_{h \downarrow 0} \frac{1}{h} \ee \left[ \left( \int_0^T F_t^h  \dd S_t \right)^2 \right] & = \frac{A}{2} \ee \Biggl[ \int_0^T S_t^2 \left(\sigma_t^2 +  \int_{\rr}  \gamma_t(z) e^{2z} \nu(\dd z)\right) \dd t \Biggr] \, .
\end{split}
\end{equation*} 
\end{theorem}
\begin{proof}
From the assumption H3-$\alpha_+$, for all $u\in \mathbb R$,
\begin{align}
|\Upsilon(u+iR)| &= \frac{1}{\bar A} \left|\int_{\mathbb R}\nu(dx) (e^{Rx-iux}-1)(e^x-1)\right|\leq C + C \left|\int_{|x|\leq 1}\nu(dx) (e^{-iux}-1)e^{Rx}(e^x-1)\right| \notag \\
&\leq C \int_{|x|\leq 1}\nu(dx)|ux|^{\alpha_+-1} |e^{-iux}-1|^{2-\alpha_+}e^{Rx}|e^x-1| \leq C(1+|u|^{\alpha_+-1}).\label{upsilonbound}
\end{align}
for some constant $C>0$ which changes from line to line. The same argument as in the proof of Theorem \ref{thm:euro_delta} then yields
\begin{align*}
I_i(t)\leq J(t) & := C\int_{\rr^2}  \frac{(1+ \sqrt{|\Re \psi(u_1)}|)(1+ \sqrt{|\Re \psi(u_2)|})}{(1+|u_1|^{(2-\alpha_+)\wedge 1})(1+|u_2|^{(2-\alpha_+)\wedge 1})} \\
&  \quad \times e^{c(\Re \psi(u_1+u_2)t + \Re {\psi}(u_1) (T-t) + \Re {\psi}(u_2) (T-t))}  \dd u_1 \dd u_2,
\end{align*}
which leads to
\begin{equation}
\int_0^T J(t) \dd t  \leq C  \int_{\rr^2}  \frac{(1+|\Re \psi(v_1)|) e^{c \Re \psi(v_1)T} }{(1+|v_1+v_2|^{(2-\alpha_+)\wedge 1})(1+|v_1-v_2|^{(2-\alpha_+)\wedge 1})}  \dd v_1 \dd v_2.\label{asin2}
\end{equation}
If $\alpha_+\leq 1$, the above expression reduces to
\begin{equation*}
\int_0^T J(t) \dd t  \leq C  \int_{\rr^2}  \frac{(1+|\Re \psi(v_1)|) e^{c \Re \psi(v_1)T} }{(1+|v_1+v_2|)(1+|v_1-v_2|)}  \dd v_1 \dd v_2,
\end{equation*} 
which is exactly the same as in Theorem \ref{thm:euro_delta}, and so the proof is completed. 

Suppose $\alpha_+>1$. By assumptions of the theorem this means that H2-$\alpha_-$ is satisfied with $\alpha_->0$. By Lemmas \ref{cond.lm} and \ref{doubleint}, the integral \eqref{asin2} then reduces to
\begin{align*}
\int_0^T J(t) \dd t  \leq C  \int_{\rr} \frac{(1+|\Re \psi(v_1)|) e^{c \Re \psi(v_1)T} }{1+|v_1|^{2\alpha_+ - 3}}  \dd v_1 \leq C  \int_{\rr} \frac{(1+|\Re \psi(v_1)|) e^{-c T |v_1|^{\alpha_-}} }{1+|v_1|^{2\alpha_+ - 3}}  \dd v_1,
\end{align*}
which is clearly finite. 
\end{proof}

\begin{theorem}[Martingale quadratic hedging, digital options, irregular regime] \label{thm:quad_irreg}
Let the pay-off function be given by $G(S_T)=1_{S_T\geq K}$ and assume \eqref{quadint.eq} for some $R>0$.
Let the hedging strategy be given by Proposition \ref{martquad.prop}.
\begin{enumerate}
\item Let $a=0$ and let $\nu$ satisfy the assumption H4-$\alpha$ with $\alpha \in \left(\frac{3}{2},2\right)$. 
Then the hedging error satisfies  
$$
\lim_{h\downarrow 0}\frac{1}{r(h)}E\left[\left(\int_0^T F^h_{t}dS_t\right)^2\right] = \frac{1}{2\pi}A Q_\alpha \frac{\gamma_+ \gamma_-}{\bar A^2} (f_+ + f_-)^{\frac{3}{\alpha}-2} p_T(\log K)
$$
with $r(h) = h^{\frac{3}{\alpha}-1},$ where $Q_\alpha$ is a constant depending only on $\alpha$ and given by
$$
Q_\alpha := \left(2\Gamma(-\alpha)\cos\left(\frac{\pi(2-\alpha)}{2}\right)\right)^{\frac{3}{\alpha}-2}\int_{\mathbb R}dv \frac{1-e^{-|v|^\alpha} - |v|^{\alpha}e^{-|v|^{\alpha}}}{|v|^{4-\alpha} (1-e^{-|v|^\alpha})},
$$
and the constants $\gamma_+,\gamma_-$ are defined in equations \eqref{upslim1}--\eqref{upslim2} in terms of $f_+$, $f_-$ and $\alpha$.
\item Let $a>0$. Then the hedging error satisfies
$$
\lim_{h\downarrow 0}\frac{1}{\sqrt{h}}E\left[\left(\int_0^T F^h_{t}dS_t\right)^2\right] = \frac{AD}{2\pi a} p_T(\log K) \frac{a^4}{\bar A^2}
$$
with $D$ as in \eqref{D.eq}. 
\end{enumerate}
\end{theorem}

\begin{remark}
The limiting case when $a=0$ and $\nu$ satisfies H4-$\frac{3}{2}$ is not covered by Theorems \ref{thm:digit_quad} and \ref{thm:quad_irreg}. While it is easy to show that the convergence rate in this case will be better than $r(h)=h^{1-\varepsilon}$ for every $\varepsilon>0$, it may not necessarily be equal to $r(h)=h$ but include, for example, additional logarithmic factors.  
\end{remark}

\begin{remark}
When $a>0$ the convergence rate is $r(h)=\sqrt{h}$ both for delta hedging and the quadratic hedging, but the corresponding constant differs by a factor $\frac{a^4}{\bar A^2}$, which is strictly smaller than one whenever the underlying  Lévy process has jumps. Therefore, also in this case the quadratic hedging strategy is superior to delta hedging. 
\end{remark}
\begin{proof}
As in the proof of Theorem \ref{thm:digit_delta}, we establish the result in two steps: first, we find an upper bound and second, we will use the dominated convergence theorem to compute the limiting renormalized hedging error. If $a>0$, we set $\alpha=2$.  

\textit{Step 1.} Let
\begin{multline*}
e_i(u_1,u_2,t):= f_i(u_1+iR,u_2+iR)\bar\phi_{T-t}(-u_1-iR)\bar\phi_{T-t}(-u_2-iR) \\ \times \phi_t(-u_1-u_2-2iR) K^{iu_1 + iu_2 - 2R} \frac{\Upsilon(u_1+iR)\Upsilon(u_2+iR)}{(R-iu_1)(R-iu_2)}.
\end{multline*}
Then, with a change of variables, 
\begin{align*}
\frac{1}{r(h)}\int_0^T(\bar \eta(t)-t) I_i(t)dt = \frac{1}{2}\int_0^T h^{1-\frac{4}{\alpha}} (\bar\eta(t)-t) \int_{\mathbb R^2} e_i\left(\frac{v_1+v_2 h^{-1/\alpha}}{2},\frac{v_1-v_2 h^{-1/\alpha}}{2},t\right)dv_1 dv_2 
\end{align*}
In this first step we would like to show that 
$$
\int_0^T h^{1-\frac{4}{\alpha}}(\bar \eta(t)-t)e_i\left(\frac{v_1+v_2 h^{-1/\alpha}}{2},\frac{v_1-v_2 h^{-1/\alpha}}{2},t\right) dt
$$
has an integrable bound. 

First, we need to analyze the behavior of $\Upsilon(u)$ as $u\to \infty$. Suppose first that $a=0$. Then 
\begin{align*}
\lim_{u\to +\infty} \frac{\Upsilon(u+iR)}{u^{\alpha-1}} &= \lim_{u\to +\infty} \frac{1}{u^{\alpha-1}\bar A}\int_{\mathbb R}(e^{Rx-iux}-1)(e^x - 1)\nu(dx) \\
&=  \lim_{u\to +\infty} \frac{1}{u^{\alpha-1}\bar A} \Biggl\{\int_{\mathbb R}(e^{-iux}-1)x\nu(dx) \\ &+ \int_{\mathbb R}e^{-iux}\{e^{(1+R)x}-e^{Rx}-x\}\nu(dx)+ \int_{\mathbb R}\{1+x-e^x\}\nu(dx)\Biggr\}.
\end{align*}
Since the two terms in the last line are bounded and $\alpha>1$, 
\begin{align*}
&\lim_{u\to +\infty} \frac{\Upsilon(u+iR)}{u^{\alpha-1}} = \lim_{u\to +\infty} \frac{1}{u^{\alpha-1}\bar A} \int_{\mathbb R}(e^{-iux}-1)x\nu(dx)\\
& \qquad= \lim_{u\to +\infty} \frac{1}{u^{\alpha-1}\bar A} \left\{\int_0^\varepsilon (e^{-iux}-1)\frac{f(x)}{x^{\alpha}}dx -  \int_0^\varepsilon (e^{iux}-1)\frac{f(-x)}{x^{\alpha}}dx\right\},
\end{align*}
where $\varepsilon$ is chosen such that $|f(x)|\leq N$ for all x with $|x|\leq \varepsilon$ and some $N<\infty$. By a change of variables and dominated convergence we then get
\begin{align*}
&\lim_{u\to +\infty} \frac{\Upsilon(u+iR)}{u^{\alpha-1}} = \lim_{u\to +\infty} \frac{1}{\bar A} \left\{\int_0^{\varepsilon u} (e^{-ix}-1)\frac{f(x/u)}{x^{\alpha}}dx -  \int_0^{\varepsilon u} (e^{ix}-1)\frac{f(-x/u)}{x^{\alpha}}dx\right\}\\
&\qquad = \frac{f_+}{\bar A}\int_0^\infty \frac{(e^{-ix}-1)}{x^\alpha}dx - \frac{f_-}{\bar A}\int_0^\infty \frac{(e^{ix}-1)}{x^\alpha}dx.
\end{align*}
Evaluating the integrals (see \cite[lemma 14.11]{sato}) and treating the limit $u\to -\infty$ in a similar manner, we finally obtain \comment{Added `,' and `.' in the equations below.}
\begin{align}
\lim_{u\to +\infty} \frac{\Upsilon(u+iR)}{|u|^{\alpha-1}}&= \frac{\Gamma(1-\alpha)}{\bar A} \{f_+ e^{-i\pi(1-\alpha)/2}-f_- e^{i\pi(1-\alpha)/2} \}:=\frac{\gamma_+}{\bar A} \, , \label{upslim1}\\
\lim_{u\to -\infty} \frac{\Upsilon(u+iR)}{|u|^{\alpha-1}}&= \frac{\Gamma(1-\alpha)}{\bar A} \{f_+ e^{i\pi(1-\alpha)/2}-f_- e^{-i\pi(1-\alpha)/2} \}:=\frac{\gamma_-}{\bar A} \, . \label{upslim2}
\end{align}
If $a>0$, a similar computation which is omitted to save space yields
\begin{align}
\lim_{u\to \infty} \frac{\Upsilon(u+iR)}{u} = -i\frac{a^2}{\bar A}.\label{upslimgauss}
\end{align}

Using property \eqref{powerbelow}, corollary \ref{repsi.bound}, estimate \eqref{fi.bound} and limits \eqref{upslim1}, \eqref{upslim2} and \eqref{upslimgauss},  
\begin{align*}
|e_i(u_1,u_2,t)| \leq C e^{-c\{t|u_1+u_2|^\alpha +(T-t)|u_1|^\alpha + (T-t)|u_2|^\alpha \}}(1+ |u_1|^{3\alpha/2 - 2})(1 + |u_2|^{3\alpha/2 - 2})
\end{align*}
for some constants $c,C>0$ which may change from line to line. By Lemma \ref{lem:growth_imp}, we then get:
$$
\left|e_i\left(\frac{v_1+v_2}{2},\frac{v_1-v_2}{2},t\right)\right| \leq C e^{-c\{T|v_1|^\alpha +(T-t)|v_2|^\alpha\}}(1+|v_1|^{3\alpha-4} +|v_2|^{3\alpha-4}),
$$
and finally, evaluating the time integral using the formula \eqref{explicit},
\begin{align*}
&\int_0^T h^{1-\frac{4}{\alpha}}(\bar \eta(t)-t)\left|e_i\left(\frac{v_1+v_2 h^{-1/\alpha}}{2},\frac{v_1-v_2 h^{-1/\alpha}}{2},t\right)\right| dt\\
&\leq C (h^{1-\frac{4}{\alpha}}(1+|v_1|^{3\alpha-4}) + h^{-2}|v_2|^{3\alpha-4}) e^{-cT|v_1|^\alpha }  \int_0^T (\bar \eta(t)-t) e^{-c(T-t)h^{-1}|v_2 |^\alpha}dt\\
& \leq C(1+|v_1|^{3\alpha-4}) e^{-cT|v_1|^\alpha } 
\frac{-c|v_2|^\alpha e^{-c |v_2|^\alpha} - e^{-c |v_2|^\alpha}+1 }{|v_2|^{4-\alpha}(1 - e^{-c |v_2|^{\alpha}})},
\end{align*}
where the last inequality follows from the bound $1-e^{-x}\leq x^{3-\frac{4}{\alpha}}$, $x\geq 0$, which holds because $3-\frac{4}{\alpha}\in(\frac{1}{3},1)$. Since the last expression is integrable with respect to $v_1$ and $v_2$ (it behaves like $\frac{1}{|v_2|^{4-2\alpha}}$ near zero and like $\frac{1}{|v_2|^{4-\alpha}}$ at infinity), step 1 is completed. 

\textit{Step 2.} Similarly to the proof of Theorem \ref{thm:digit_delta}, we compute 
\begin{align*}
&\lim_{h\to 0} \int_0^T h^{1-\frac{4}{\alpha}}(\bar \eta(t)-t)\{e_2+e_4\}\left(\frac{v_1+v_2 h^{-1/\alpha}}{2},\frac{v_1-v_2 h^{-1/\alpha}}{2},t\right) dt\\
& = e^{T\psi(-v_1 - 2iR)}K^{iv_1 - 2R} L_1 L_2 L_3,
\end{align*}
where $L_1$ and $L_2$ are the same as in the proof of Theorem \ref{thm:digit_delta}, and 
$$
L_3 = \lim_{h\to 0} h^{2-\frac{4}{\alpha}} \frac{\Upsilon\left(\frac{v_1 + v_2 h^{-1/\alpha}}{2}+iR\right)\Upsilon\left(\frac{v_1 - v_2 h^{-1/\alpha}}{2}+iR\right)}{\left(R-i\frac{v_1 + v_2 h^{-1/\alpha}}{2}+iR\right)\left(R-i\frac{v_1 - v_2 h^{-1/\alpha}}{2}\right)} =\gamma_+ \gamma_- \left(\frac{v_2}{2}\right)^{2\alpha-4}
$$
if $a=0$
and $L_3 = \frac{a^4}{\bar A^2}$ if $a>0$.
Assembling the three factors together, the proof is completed.
\end{proof}

\section*{Acknowledgements}
This research was supported by the Chair Financial Risks of the Risk Foundation sponsored by Société Générale, the Chair Derivatives of the Future sponsored by the Fédération Bancaire Française, the Chair Finance and Sustainable Development sponsored by EDF and Calyon, and the Royal Physiographic Society in Lund.


\appendix

\small
\section{Characteristic function estimates in exponential Lévy models and other useful results}
Below we use the common notation introduced in the beginning of section \ref{pricing.sec}.
\begin{lemma}\label{cond.lm}${}$
\begin{enumerate}
\item Let the Lévy measure $\nu$ satisfy the assumption H1 on page \pageref{H1}. Then (i) for every $t>0$ there exist constants $C>0$ and $c>0$ such that 
$$
|\phi_t(z)| \leq C|z|^{-c},\quad z\in \mathbb R
$$
and (ii) there exists a constant $c$ such that 
\begin{equation}
u \geq v \quad \text{implies} \quad \Re \psi(u) \leq \Re \psi(v)+c \quad \text{for all} \quad u , v >0 \, . \label{growth.assum}
\end{equation}
\item Let $\nu$ satisfy the assumption H2-$\alpha$ with $\alpha\in(0,2)$. Then there exist $c>0$ and $C>0$ such that
$$
|\phi_t(z)|\leq C e^{-ct|z|^\alpha}, \quad \forall t>0, \forall z.
$$ 
\item Let $\nu$ satisfy the assumption H4-$\alpha$ with $\alpha\in(1,2)$ and let $a=0$. Then the characteristic exponent $\psi$ satisfies
\begin{align}
\lim_{u\to +\infty} \frac{\psi(u)}{|u|^{\alpha}}=-c_+ \quad \text{and}\quad \lim_{u\to -\infty} \frac{\psi(u)}{|u|^{\alpha}}= -c_-,\label{psilim}
\end{align}
where 
\begin{align*}
c_+= -\Gamma(-\alpha) \{f_+ e^{-i\pi\alpha/2}+f_- e^{i\pi\alpha/2} \},\qquad c_- = -\Gamma(-\alpha) \{f_+ e^{i\pi\alpha/2}+f_- e^{-i\pi\alpha/2} \}
\end{align*}
and there exist constants $c_1,c_3\in \mathbb R$ and $c_2,c_4 >0$ such that for all $u\in \mathbb R$, 
\begin{align}
c_1 - c_2 |u|^\alpha< \Re\psi(u)<c_3 - c_4 |u|^\alpha.\label{powerbelow}
\end{align}
If $a>0$ then equations \eqref{psilim} hold with $c_+=c_- = \frac{a^2}{2}$ and inequality \eqref{powerbelow} holds with $\alpha=2$.
\end{enumerate}
\end{lemma}

\begin{proof} ${}$
\begin{enumerate}
\item The property (i) is Lemma 28.5 in \cite{sato}; let us concentrate on property (ii). Since this property is linear in $\psi$ and clearly satisfied by a Lévy process with zero Lévy measure, we can suppose without loss of generality that $a=0$. 
 Let $u \geq v$ and $u, v > 0$, then
\begin{align}
\Re \psi(u)-\Re \psi(v) & = \left\{ \int_{\rr} (\cos(ux)-1) \nu_0(d x)- \int_{\rr} (\cos(vx)-1) \nu_0(d x) \right \} \notag \\
& + \left\{ \int_{\rr} (\cos(ux)-1) \frac{k(x)}{|x|} d x - \int_{\rr} (\cos(vx)-1) \frac{k(x)}{|x|} d x \right\} \, . \label{prop:proof:1}
\end{align}
A change of variables ($y=ux$ and $y=vx$) yields for the second term
\begin{align*}
\int_{\rr} (\cos(ux)-1) \frac{k(x)}{|x|} d x - \int_{\rr} (\cos(vx)-1) \frac{k(x)}{|x|} d x  
= \int_{\rr} \frac{(\cos(y)-1)}{|y|} (k(y/u)-k(y/v)) d y  \leq 0 .
\end{align*}
Thus, \eqref{growth.assum} is satisfied with $c=2 \int_{\rr} \nu_0(d x)$.
\item This follows from the proof of Proposition 28.3 in \cite{sato}.
\item Choose $\varepsilon>0$ and $N<\infty$ such that $|f(x)|\leq N$ for all $x$ with $|x|\leq \varepsilon$. Since $\alpha>1$,
\begin{align*}
\lim_{u\to+\infty} \frac{\psi(u)}{u^\alpha} &= \lim_{u\to +\infty} \frac{1}{u^\alpha} \int_{|x|\leq \varepsilon} (e^{iux}-iux-1)\nu(dx) \\ &= \lim_{u\to +\infty} \frac{1}{u^\alpha} \int_{|x|\leq \varepsilon} (e^{iux}-iux-1)\frac{f(x)}{|x|^{1+\alpha}}dx.
\end{align*}
By a change of variables and dominated convergence we then get
\begin{align*}
\lim_{u\to+\infty} \frac{\psi(u)}{u^\alpha}& = \lim_{u\to +\infty}  \int_{|x|\leq \varepsilon u} (e^{ix}-ix-1)\frac{f(x/u)}{|x|^{1+\alpha}}dx\\
&= f_- \int_{-\infty}^0 (e^{ix}-ix-1)\frac{dx}{|x|^{1+\alpha}} + f_+ \int_0^{+\infty} (e^{ix}-ix-1)\frac{dx}{|x|^{1+\alpha}}. 
\end{align*}
These integrals are explicitly computed in \cite[page 84]{sato}, and the case $u\to -\infty$ can be treated in a similar manner.  The estimates \eqref{powerbelow} follow directly from \eqref{psilim}. For the case $a>0$ see \cite[page 16]{bertoin}. 
\end{enumerate}
\end{proof}

\begin{lemma}\label{doubleint}
Let $\alpha >\frac{1}{2},\quad \alpha\neq 1$. Then there exists $C<\infty$ with
\begin{align*}
\int_{\mathbb R} \frac{dv}{(1+|u+v|^\alpha)(1+|u-v|^\alpha)} \leq C(1+|u|)^{1-2\alpha}.
\end{align*}
In the case $\alpha = 1$, there exits $C<\infty$ with
\begin{align*}
\int_{\mathbb R} \frac{dv}{(1+|u+v|)(1+|u-v|)} \leq C\frac{1+\log(1+|u|)}{1+|u|}.
\end{align*}
\end{lemma}
\begin{proof}
In the case $\alpha\neq 1$, we have
\begin{align*}
&\int_{\mathbb R} \frac{dv}{(1+|u+v|^\alpha)(1+|u-v|^\alpha)} = \int_{0}^\infty \frac{dv}{(1+||u|+v|^\alpha)(1+||u|-v|^\alpha)} \\
& \qquad \leq \frac{2}{1+|u|^\alpha}\int_0^{2|u|} \frac{du}{1+||u|-v|^\alpha} + 2 \int_{2|u|}^\infty \frac{dv}{(1+|v|^\alpha)(1+|\frac{v}{2}|^\alpha)}\\
& \qquad \leq \frac{C}{(1+|u|)^\alpha}\int_0^{|u|}\frac{dv}{(1+v)^\alpha} + C \int_{2|u|}^\infty \frac{dv}{(1+v)^{2\alpha}} \leq \frac{C}{(1+|u|)^{2\alpha-1}},
\end{align*}
where $C$ is a constant which may change from inequality to inequality. In the case $\alpha=1$ the proof is done in a similar manner (the logarithmic factor appears in the first integral of the last line). 
\end{proof}

\begin{lemma} \label{lem:growth_imp}${}$
\begin{enumerate}
\item For any Lévy process $X$,
\begin{equation}
\Re \psi(u) \leq \frac{1}{4} \Re \psi(2u),\quad u\in\mathbb R.\label{eqn:lem_growth_prop_2}
\end{equation}
\item Assume that there exists a constant $C>0$ such that
\begin{equation}
u \geq v \quad \Rightarrow \quad \Re\psi(u) \leq \Re\psi(v) + C \, , \quad u,v > 0 \, .
\label{eqn:lem_growth_prop_1}
\end{equation}
Then
\begin{align}
& \Re\psi((u+v)/2) + \Re\psi((u-v)/2) \leq (\Re\psi(u)+\Re\psi(v))/8 + C/4  \label{eqn:lem_H3}
\intertext{and}
& \sqrt{|\Re\psi((u+v)/2)||\Re\psi((u-v)/2)|} \leq 8 (|\Re\psi(u/2)|+|\Re\psi(v/2)|) + 2C  \label{eqn:lem_H4}
\end{align}
for all $u,v\in \mathbb R$. 
\end{enumerate}
\end{lemma}
\begin{proof}
By the Lévy-Khintchine formula
\begin{align*}
\Re \psi(2 u) & = -4 a^2 \frac{u^2}{2} + \int_{\rr}(\cos(2 u x)-1) \nu(\dd x) \\
& = -4 a^2 \frac{u^2}{2} + 2 \int_{\rr} (\cos(u x)-1)^2 \nu(\dd x) + 4 \int_{\rr} (\cos(u x)-1) \nu(\dd x) \\
& \geq -4 a^2 \frac{u^2}{2} + 4 \int_{\rr} (\cos(u x)-1) \nu(\dd x) \, ,
\end{align*}
which proves \eqref{eqn:lem_growth_prop_2}. Combined with \eqref{eqn:lem_growth_prop_1}, this immediately yields
\begin{align}
\Re\psi((u+v)/2) \leq  (\Re\psi(u) + \Re\psi(v))/8 + C/4 \, ,\quad u,v >0 \, , \label{eqn:lem_H1}
\end{align}
and therefore, since $\Re\psi \leq 0$, for all $u,v \in \mathbb R$,
\begin{multline*}
\Re\psi((u+v)/2) + \Re\psi((u-v)/2) = \Re\psi(|u+v|/2) + \Re\psi(|u-v|/2) \\
= \Re\psi((|u|+|v|)/2) + \Re\psi((|u|-|v|)/2) \leq (\Re\psi(|u|) + \Re\psi(|v|))/8 + C/4 \, .
\end{multline*}
Finally, taking absolute values in the above inequality, we have
\begin{align*}
|\Re\psi((u+v)/2)| + |\Re\psi((u-v)/2)| \geq (|\Re\psi(|u|)| + |\Re\psi(|v|)|)/8 - C/4,
\end{align*}
and after a change of variables,
$$
|\Re\psi((u+v)/2)| + |\Re\psi((u-v)/2)| \leq 8 (|\Re\psi(u/2)|+|\Re\psi(v/2)|)+2C,
$$
from which \eqref{eqn:lem_H4} follows. 
\end{proof}

\begin{lemma} \label{lem:UVpsi}
Let $R, R'\in \mathbb R$ with $R\leq R'$ and assume
$$
\int_{|x|>1}e^{-xR }\nu(dx)<\infty,\quad\text{and}\quad \int_{|x|>1}e^{-xR' }\nu(dx)<\infty. 
$$
Then there exists $C>0$ such that for all $u,v \in \mathbb C$ with $\Im u \in [R,R']$, $\Im v \in [R,R']$ and $\Im u + \Im v\in [R,R']$, 
\begin{equation*}
|\psi(u+v)-\psi(u)-\psi(v)| \leq C(1+\sqrt{|\Re \psi(\Re u)|})(1 +  \sqrt{|\Re \psi(\Re v)|}).
\end{equation*}
\end{lemma}
\begin{proof}
From the Lévy-Khintchine formula,
\begin{align}
&|\psi(u+v)-\psi(u)-\psi(v)|  = \left| -a^2uv + \int_{\rr} (e^{\ii u x}-1) (e^{\ii v x}-1)  \nu(\dd x)  \right| \notag\\
& \leq  c_1 + c_2 a^2 (1+|\Re u|)(1+(|\Re v|) + \left( \int_{|x|\leq 1} |e^{\ii u x}-1|^2 \nu(\dd x)  \right)^{\frac{1}{2}} \left( \int_{|x|\leq 1} |e^{\ii v x}-1|^2  \nu(\dd x)  \right)^{\frac{1}{2}}.\label{bound3psi}
\end{align}
for some constants $c_1$ and $c_2$. 
Let $u=\alpha + i\beta$, then
\begin{align}
\Re \psi(\Re u) = -\frac{a^2\alpha^2}{2} - \int_{\mathbb R}(\cos \alpha x -1)\nu(dx). \label{repsi}
\end{align}
Therefore,
\begin{multline}
\int_{|x|\leq 1} |e^{iux}-1|^2 \nu(\dd x) \\
= \int_{|x|\leq 1} (e^{-\beta x}-1)^2 \nu(\dd x) + 2 \int_{|x|\leq 1} e^{-\beta x} (1-\cos(\alpha x)) \nu( \dd x) \leq c_3 + 2 e^{|\beta|} |\Re \psi (\alpha)|\label{boundsquare}
\end{multline}
for some $c_3<\infty$. Combining \eqref{bound3psi}, \eqref{repsi} and \eqref{boundsquare}, the proof is completed. 
\end{proof}

\begin{corollary}\label{repsi.bound}
Let $R$ and $R'$ be as in Lemma \ref{lem:UVpsi}. 
Then there exist constants $C_1\in \mathbb R$ and $C_2 >0$ for all $u\in \mathbb C$ with $\Im u \in [R,R']$, 
$$
\Re \psi(u) \leq C_1 + C_2 \Re \psi(\Re u). 
$$ 
\end{corollary}

\begin{lemma}\label{lem:PQUVpsi}
Let $R, R', \bar R, \bar R'\in \mathbb R$ with $R\leq R'$ and $\bar R\leq \bar R'$, such that 
\begin{align*}
&\int_{|x|>1}(e^{-xR }+e^{-xR'})\nu(dx)<\infty,\quad\text{and}\quad &\int_{|x|>1}(e^{-x\bar R }+e^{-x\bar R' })\bar \nu(dx)<\infty.
\end{align*}
Then there exists $C>0$ such that for all $u,v\in \mathbb C$ with $\Re u = \Re v$, $\Im u \in[R,R']$ and $\Im v \in [\bar R,\bar R']$,
$$
|\psi(u)-\bar{\psi}(v)| \leq C(1+ \sqrt{|\Re \psi(\Re u)|}).
$$
\end{lemma}
\begin{proof}
Let $z = \Re u = \Re v$. The difference in question can be rewritten as
\begin{align}
\psi(u)-\bar{\psi}(v) &= \psi(z)-\bar{\psi}(z) \notag\\ &+ \psi(u)-\psi(\Re u) - \psi (i \Im u) \notag\\ &+ \bar\psi(\Re v) + \bar\psi(i \Im v) - \bar\psi(v) \notag\\ &+ \psi(i \Im u) - \psi(i \Im v).\label{decompos}
\end{align}
Let us start with the first line.  
From Theorem 33.1 in~\cite{sato}, $\gamma-\bar \gamma-\int_{-1}^{1} x (\nu-\bar{\nu}) (\dd x) = a^2\eta$ for some $\eta\in \mathbb R$. This relation yields
$$
|\psi(z)-\bar{\psi}(z)| = \left| a^2\eta z + \int_{\rr} (e^{\ii z x}-1) (e^{\varphi(x)}-1)  \nu(\dd x)   \right|  \, ,
$$
where $\varphi(x)$ is defined by $e^{\varphi(x)}=\nu(\dd x)/\bar{\nu}(\dd x)$. Equation \eqref{repsi} shows that when $a>0$, the first term in the right-hand side satisfies the required bound; let us focus on the second term (the integral). In the following, $C$ denotes a constant which may change from line to line.
\begin{align*}
&\left|\int_{\rr} (e^{\ii z x}-1) (e^{\varphi(x)}-1)  \nu(\dd x)   \right| \leq C + \left|\int_{|x|\leq 1} (e^{\ii z x}-1) (e^{\varphi(x)}-1)  \nu(\dd x)\right| \\
&\leq C + \left|\int_{\{x:|x|\leq 1\}\cap \{x:|\varphi(x)|\leq 1\}} (e^{\ii z x}-1) (e^{\varphi(x)}-1)  \nu(\dd x)\right| \\ 
&\leq C + \left(\int_{\{x:|x|\leq 1\}} |e^{\ii z x}-1|^2   \nu(\dd x)\right)^{1/2} \left(\int_{\{x:|\varphi(x)|\leq 1\}} (e^{\varphi(x)}-1)^2  \nu(\dd x)\right)^{1/2}\\
&\leq C + C \left(\int_{\{x:|x|\leq 1\}} |e^{\ii z x}-1|^2   \nu(\dd x)\right)^{1/2} \leq C(1+ \sqrt{|\Re \psi(z)|}), 
\end{align*}
where the second and the fourth inequality follow from \cite[Remark 33.3 ]{sato} and the last one follows from \eqref{boundsquare}. Finally,
\begin{align}
|\psi(z)-\bar{\psi}(z)| \leq C(1+ \sqrt{|\Re \psi(z)|}).\label{diffequiv}
\end{align}
Applying Lemma \ref{lem:UVpsi} to the second and the third line in \eqref{decompos}, and observing that the fourth line is bounded by a constant, we get
$$
|\psi(u)-\bar{\psi}(v)| \leq C\left(1+ \sqrt{|\Re \psi(z)|} + \sqrt{|\Re \bar\psi(z)|}\right).
$$
Now, using \eqref{diffequiv} for a second time, the proof is completed.\comment{Removed a couple of blank lines as well as added a full stop in the end of the sentence.}
\end{proof}

\section{Martingale representations for Fourier integrals}

\begin{lemma}\label{martrep.prop}
Let the process $F$ be defined by \comment{Added `,' and `.' in some of the equations below.}
\begin{align}
F_t = \int_{\mathbb R} f(u)\bar\phi_{T-t}(-u-iR)S_t^{R'-iu}du,\label{strategy.eq}
\end{align}
where $f$ is a deterministic function satisfying
\begin{align}
\int_{\mathbb R}|f(u)\bar\phi_{T-t}(-u-iR)|du<\infty,\quad \forall t<T.\label{martcond.eq}
\end{align}
Assume
\begin{align}
&\int_{|x|>1}e^{2R'x}\nu(dx)<\infty\quad \text{and}\quad \int_{|x|>1}e^{Rx}\bar \nu(dx)<\infty.\label{histint.eq}
\end{align}
Then the representation \eqref{levyito} holds for $F$ with
\begin{align}
\mu_t &=  \int_{\mathbb R}f(u)\bar\phi_{T-t}(-u-iR)S_t^{R'-iu}(\psi(-u-iR')-\bar\psi(-u-iR))du,\\
\sigma_t &= {a} \int_{\mathbb R} f(u)\bar\phi_{T-t}(-u-iR)(R'-iu)S_t^{R'-iu}du,\\
\gamma_t(z)&=  \int_{\mathbb R} f(u)\bar\phi_{T-t}(-u-iR) S_t^{R'-iu}(e^{(R'-iu)z}-1)du.
\end{align}
\end{lemma}
\begin{proof}
Let $t<T$. Applying the Itô formula under the integral sign in \eqref{strategy.eq}, we find, under the condition \eqref{histint.eq},
\begin{align}
F_t - F_0 &=  \int_{\mathbb R} du f(u) \int_0^t \bar\phi_{T-s}(-u-iR) S_s^{R'-iu}(\psi(-u-iR')-\bar\psi(-u-iR))ds
\notag\\ &+  \int_{\mathbb R} du f(u) \int_0^t \bar\phi_{T-s}(-u-iR) (R'-iu)  S_s^{R'-iu} a dW_s\notag\\
& +  \int_{\mathbb R} du f(u) \int_0^t  \bar\phi_{T-s}(-u-iR)  S_{s-}^{R'-iu}\int_{\mathbb R}(e^{(R'-iu)z}-1)\tilde J_X(ds\times dz).\label{itoprice}
\end{align}
To finish the proof, we apply a suitable Fubini-type theorem to each of the three terms. For the first term, we use the standard Fubini theorem (path by path), whose applicability condition is
\begin{align*}
&\int_{\mathbb R} du \int_0^t \left|f(u)  \bar\phi_{T-s}(-u-iR) S_s^{R'-iu}(\psi(-u-iR')-\bar\psi(-u-iR))\right|ds\\
& \leq C\sup_{s\leq t}S_s^{R'} \int_{\mathbb R} du \int_0^t \left|f(u)  e^{(T-s)\Re\bar\psi(-u-iR)}\right|\left(1+\sqrt{|\Re\bar\psi(u)|}\right)ds\\
& \leq C\sup_{s\leq t}S_s^{R'} \int_{\mathbb R} du \int_0^t \left|f(u)  e^{(T-s)c\Re\bar\psi(u)}\right|\left(1+\sqrt{|\Re\bar\psi(u)|}\right)ds\\
&\leq C\sup_{s\leq t}S_s^{R'} \int_{\mathbb R} du \left|f(u)  e^{(T-t)c\Re\bar\psi(u)}\right|<\infty \quad a.s.,
\end{align*}
where we used Lemma \ref{lem:PQUVpsi} to pass from line 1 to line 2 and Corollary \ref{repsi.bound} from line 2 to line 3, and the constants $c>0$ and $C>0$ may change from line to line. Note that $\sup_{s\leq t}S_s^{R'}<\infty$ a.s. because $S$ is càdlàg.

Let us now assume that $\sigma>0$ and study the second term in the right-hand side of \eqref{itoprice}, which can be written as \comment{Added `,' and `.' in some of the equations below.}
\begin{align}
\int_{\mathbb R} \mu(du) \int_0^t H^u_s dW_s,\label{fubini}
\end{align}
where $\mu(du) = |f(u)\bar\phi_{T-t}(-u-iR)| du$ is a finite positive measure on $\mathbb R$
and 
$$
H^u_s = \frac{a f(u) \bar\phi_{T-s}(-u-iR)}{2\pi |f(u)\phi_{T-t}(-u-iR)|}(R'-iu) S_s^{R'-iu}.
$$
By the Fubini theorem for stochastic integrals (see \cite[page 208]{protter2nd}), we can interchange the two integrals in \eqref{fubini} provided that
\begin{align}
E\int_0^t \mu(du) |H^u_s|^2 ds < \infty.\label{fubinicond} 
\end{align}
From Corollary \ref{repsi.bound} it follows that 
$$
\frac{|\bar\phi_{T-s}(-u-iR)|}{|\bar\phi_{T-t}(-u-iR)|} \leq C,
$$
for all $s\leq t\leq T$ for some constant $C>0$ which does not depend on $s$ and $t$. To prove \eqref{fubinicond} it is then sufficient to check
$$
E \int_0^t \int_{\mathbb R}|f(u)\bar\phi_{T-t}(-u-iR)| | S_s^{2(R'-iu)}|^2 (R'-iu)^2 du dt<\infty.
$$
After evaluating the expectation explicitly using \eqref{histint.eq}, the finiteness of this integral follows from
\begin{align}
|\bar\phi_{T-t}(-u-iR)|\leq Ce^{-(T-t)\frac{\sigma^2 u^2}{2}}.\label{gaussbound}
\end{align}
Therefore, the second term on the right-hand side of \eqref{itoprice} is equal to $\int_0^t \sigma_s dW_s$.
 
Let us now turn to the third term in the right-hand side of \eqref{itoprice}. Here we need to apply the Fubini theorem for stochastic integrals with respect to a compensated Poisson random measure \cite[Theorem 5]{appelbaum.06} and the applicability condition boils down to
\begin{align}
E\int_0^t \int_{\mathbb R}|f(u)\bar\phi_{T-t}(-u-iR)| | S_s^{2(R'-iu)}|^2 \int_{\mathbb R} |e^{(R'-iu)z}-1|^2\nu(dz) du dt<\infty.\label{integrgamma}
\end{align}
If $\sigma>0$, this is once again guaranteed by \eqref{gaussbound}, and when $\sigma=0$, 
$$
\int_{\mathbb R} |e^{(R'-iu)z}-1|^2\nu(dz) = \psi(-2iR') - 2 \Re \psi(-u-iR').
$$
Therefore, evaluating the expectation explicitly, and using Lemma \ref{lem:PQUVpsi}, the integrability condition \eqref{integrgamma} reduces to
$$
\int_{\mathbb R}|f(u)\bar\phi_{T-t}(-u-iR)| (1+|\Re\bar\psi(u)|) du <\infty,
$$
which holds by Corollary \ref{repsi.bound} and assumption \eqref{martcond.eq}.  
\end{proof}

\end{document}